\documentclass[11pt]{amsart}
\usepackage{fullpage} 
\usepackage{algorithm}
\usepackage{algorithmic}
\usepackage{graphicx}
\usepackage{amsfonts}
\usepackage{amsmath}
\usepackage{url}

\newcommand{\amenable}{amenable}
\newcounter{foo}
\newtheorem{theorem}[foo]{Theorem}
\newtheorem{lemma}[foo]{Lemma}

\newtheorem{proposition}{Proposition}[section]
\newtheorem{claim}[foo]{Claim}

\begin{document}

\title{\Large Wireless Connectivity and Capacity}
\author[M. Halld\'orsson]{Magn\'us M. Halld\'orsson}
\address[M. Halld\'orsson]{School of Computer Science\\
Reykjavik University\\
Reykjavik 101, Iceland}
\email{mmh@ru.is}

\author[P. Mitra]{Pradipta Mitra}
\address[P. Mitra]{School of Computer Science\\
Reykjavik University\\
Reykjavik 101, Iceland}
\email{ppmitra@gmail.com}
\date{}

\maketitle


\begin{abstract} \small\baselineskip=9pt
Given $n$ wireless transceivers located in a plane, a fundamental
problem in wireless communications is to construct a strongly
connected digraph on them such that the constituent links can be
scheduled in fewest possible time slots, assuming the SINR model of
interference.

In this paper, we provide an algorithm that connects an arbitrary
point set in $O(\log n)$ slots, improving on the previous best bound
of $O(\log^2 n)$ due to Moscibroda.  This is complemented with a
super-constant lower bound on our approach to connectivity.
An important feature is that the algorithms allow for bi-directional
(half-duplex) communication.

One implication of this result is an improved bound of $\Omega(1/\log
n)$ on the worst-case capacity of wireless networks, matching the best
bound known for the extensively studied average-case.

We explore the utility of oblivious power assignments, and show that 
essentially all such assignments result in a worst case bound of
$\Omega(n)$ slots for connectivity.
This rules out a recent claim of a $O(\log n)$ bound using oblivious power. 
On the other hand, using our result we show that 
$O(\min(\log \Delta, \log n \cdot (\log n + \log \log \Delta)))$ slots suffice, where $\Delta$ is
the ratio between the largest and the smallest links in a minimum spanning tree of the points. 

Our results extend to the related problem of minimum latency
aggregation scheduling, where we show that aggregation scheduling with
$O(\log n)$ latency is possible, improving upon the previous best
known latency of $O(\log^3 n)$.
We also initiate the study of network design problems in the SINR
model beyond strong connectivity, obtaining similar bounds 
for biconnected and $k$-edge connected structures.
\end{abstract}

\section{Introduction}
A key architectural goal in wireless adhoc networks is to ensure that
each node in the network can communicate with every other node
(perhaps by routing through other nodes). This requires that the nodes
be connected through a communication overlay.  The problem can be
abstracted as such: Given $n$ points on the plane (each representing a
wireless node), how \emph{efficiently} can one ensure connectivity
among the points?

The notion of efficiency in a wireless setting is crucially dependent on that distinguishing feature
of wireless networks: interference. Two or more simultaneous communications in the
same wireless channel interfere with each other, potentially destroying all or some of the communications.
Thus, easy as it might be to come up with a set of links (a link is an directed edge between two nodes) that connect the $n$ nodes, it is highly 
unclear whether or not one can schedule these links in a small amount of time. 
This fundamental problem has been
the focus of substantial amount of research \cite{MoWa06,moscibroda06b,Moscibroda07,kumar2005,Dousse03impactof}.

The model of interference is of course a crucial
aspect. Traditionally, all theoretical results have been in graph-based
models, with either fixed radii (unit-disc graphs and quasi-unit
disc graphs) and variable radii (geometric radio networks and
protocol model), 
while engineering research has focused on largely non-algorithmic 
studies in more complex models.
We adopt the SINR (signal to noise and interference ratio) model,
a.k.a.\ the physical model, of interference. The main differences are
two-fold: the received signal is a decaying function of distance (rather than
being on/off), and interferences from multiple transmitters sum up.
While more involved analytically, the SINR models is known to
be more realistic than graph-based ones, as shown theoretically as well as
experimentally~\cite{GronkMibiHoc01,MaheshwariJD08,Moscibroda2006Protocol}.

The first worst-case guarantee for wireless connectivity 
in the SINR model was provided by Moscibroda and
Wattenhofer \cite{MoWa06}, who showed how to construct a strongly
connected set of links that can 
be scheduled in $O(\log^4 n)$ slots. This was improved to $O(\log^3
n)$ in \cite{moscibroda06b} and finally 
to $O(\log^2 n)$ by Moscibroda \cite{Moscibroda07}, which is the 
best bound currently known.


Our main result is the following: Any minimum spanning tree
(arbitrarily oriented) on $n$ nodes on the plane can be scheduled in
$O(\log n)$ slots. This immediately leads to a $O(\log n)$ worst-case
bound for strong connectivity, by orienting the tree towards an arbitrary
root and then using the same tree with the orientation reversed.
Thus we improve the connectivity bound by a $\log n$ factor, while 
giving at the same time a simple characterization of the resultant network in
terms of the natural MST structure.

The connectivity problem is closely related to the capacity of a
wireless network, a subject of a vast literature.
The computational throughput capacity of a network is the sustained
rate at which data can be aggregated to an information sink, which is
really the \emph{raison d'\^etre} of wireless sensor networks.  
At each time step, data is introduced at each source node.
If the aggregation function is compressible, like sum or max, 
only one item of data needs to be forwarded on each link.
A short schedule that is repeated as needed yields high throughput using buffering.  
Bounds for the connectivity problem lead therefore immediately to
equivalent bounds for worst-case capacity of wireless network (for
compressible functions) \cite{Moscibroda07}.

Indeed, this particular application also 
highlights the specific benefits of adopting the SINR model.  The best
known bound on the average-case capacity in the SINR model is
$\Omega(1/\log n)$, given in the influential work of Gupta and Kumar
\cite{Kumar00}.  On the other
hand, whereas the average case throughput capacity in the protocol model is
$\Theta(1/\log n)$ \cite{Kumar00}, the worst-case capacity
is only $\Theta(1/n)$ \cite{Moscibroda07}.

We also study a variation of the connectivity problem inspired by the
sensor networking application mentioned above is known as
\emph{minimum-latency aggregation scheduling}. In this variation, one
seeks a tree aggregating to a information sink (as before), but with
the additional requirement that links must be scheduled after all
links below them in the tree are scheduled.  
A straightforward modification of our algorithm achieves this in
optimal $O(\log n)$ slots, improving on the $O(\log^3 n)$ result previously known \cite{Li:2010:MAS:1868521.1868581}.

We conjecture that a logarithmic bound is necessary for connectivity.
One reason is that it matches the average-case bound, which has been a
highly researched topic \cite{Kumar00}.
We also give a construction that shows that our approach cannot yield a constant upper bound.
It is distinguished from all previous lower bound constructions in the SINR model in that it is (necessarily) not based on showing that pairs of links are incompatible. Without being able to show the existence of large ``cliques'', hardness results in the SINR with power control are hard to come by.

An important -- and perhaps surprising -- feature of our method is that it allows for bidirectional
communication. Namely, the links scheduled in each slot can
communicate in either direction without affecting or being affected by 
the other scheduled links.
This is important in a communication
setting because of the need to supply acknowledgements and flow
control, and is sometimes viewed as indispensable.
The previously studied algorithms 
\cite{MoWa06,moscibroda06b,Moscibroda07} all assumed unidirectional
communication. In fact, it was taken for granted that
unidirectionality could not be avoided, sometimes with references to 
lower bounds from graph-based models \cite{moscibroda06b}. Our algorithm uses
different power for the two directions of each link; we show that
to be unavoidable by constructing instances for which 
the use of symmetric power on bidirectional links 
forces the use of $\Omega(n)$ slots.

Power assignments are yet another important issue in wireless protocols.
It is preferable if power settings are locally computable.
A power assignment is \emph{oblivious} if it depends only on the
length of the respective link.
Recently, a $O(\log n)$-slot connectivity
algorithm was claimed that used a particular oblivious power
assignment \cite{DBLP:conf/wdag/KowalskiR10}.
Unfortunately, there are problems with the proof (specifically, in Lemma 5 whose proof is not in the conference version) as acknowledged by one of the authors \cite{Kowalski11}. 
Actually, that general approach is bound to fail;
namely, we show that essentially all oblivious power
assignments (including the one used in \cite{DBLP:conf/wdag/KowalskiR10})
require $\Omega(n)$ slots in the worst case.

On the other hand, when the edge lengths in the MST differ by a factor
of at most $\Delta$, then combining the results here with recent work
\cite{SODA11} gives a $O(\log n (\log\log \Delta + \log n))$ slot
connectivity algorithm that uses a certain oblivious power assignment
called \emph{mean power}.


We use our approach as a starting point for the first excursion into
network design problems beyond strong connectivity.
By applying the connectivity routine a constant number of times, we
find that we can solve other connectivity problems with asymptotically
the same number of slots, including biconnectivity and $k$-edge
connectivity.

\paragraph{Outline of the paper.}
We introduce the SINR model and related notation in
Sec.~\ref{sec:model}, followed by quick overview of related work in Sec.~\ref{sec:related}.
The connectivity algorithm is given in Sec.~\ref{sec:conn},
with a subsection on a limitation result.
We extend the method to a bi-directional model
of communication in Sec.~\ref{sec:bidi}, and examine the power of oblivious
power in Sec.~\ref{sec:oblivious}.
Extensions to other connectivity problems are treated in Sec.~\ref{sec:design}.

\section{Model and Preliminaries.}
\label{sec:model}

Given is a set $P = \{p_1, p_2, \ldots, p_n\}$ of points on the Euclidean plane.
A link $\ell = (s, r)$ is a directed edge from point $s$ (the ``sender") to point $r$ (the ``receiver").
The goal is compute a set of links that strongly connect $P$ and to schedule them
in $O(\log n)$ slots.

The distance between two points $x$ and $y$ is denoted $d(x,y)$. The asymmetric distance from link $\ell = (s, r)$ to link $\ell' = (s', r')$ is the distance from
$\ell$'s sender to $\ell'$'s receiver, denoted $d_{\ell \ell'} = d(s, r')$. The length of link $\ell$ is denoted  simply $\ell$.
For a link set $L$, let $\Delta$ denote the ratio between the maximum and minimum length of a link in $L$.

When a point $s$ transmits as a sender of link $\ell$, it uses some transmission power $P_{\ell}$.
We adopt the \emph{physical model} (or \emph{SINR model})
of interference: a communication over a link $\ell = (s, r)$ succeeds if and only if the
following condition holds:
\begin{equation}
 \frac{P_{\ell}/\ell^\alpha}{\sum_{\ell' \in S \setminus  \{\ell\}}
   P_{\ell'}/d_{\ell' \ell}^\alpha + N} \ge \beta, 
 \label{eq:sinr}
\end{equation}
where $\alpha > 2$ is the path loss constant, $N$ is a universal constant denoting the ambient noise, $\beta \ge 1$ denotes the minimum
SINR (signal-to-interference-noise-ratio) required for a message to be successfully received,
and $S$ is the set of concurrently scheduled links in the same \emph{slot} with $\ell$.
We say that $S$ is \emph{SINR-feasible} (or simply \emph{feasible}) if (\ref{eq:sinr}) is
satisfied for each link in $S$. 

We will use the notion of affectance of \cite{HW09}, as refined in \cite{KV10}, 
which is a scaled interference measure from one link on another, defined as 
  \[ a_{\ell}(\ell') 
     = \min\left\{1, c_{\ell'} \frac{P_{\ell}/d_{\ell \ell'}^\alpha}{P_{\ell'}/\ell'^\alpha}\right\}
  \] 
where $c_{\ell'} = \beta/(1 - \beta N {\ell'}^\alpha/P_{\ell'})$ is a constant
depending only on the length and power of the link $\ell'$. 
As in previous work \cite{MoWa06,moscibroda06b,Moscibroda07,KesselheimSoda11}, we assume that
powers can be scaled up as needed, which implies that the effect of the noise $N$ (and the coefficient $c_{\ell'}$) can be ignored.
It holds that $\ell'$ is feasible in $S$ iff
\begin{equation}
\sum_{\ell \in S} a_{\ell}(\ell') \leq 1, 
\end{equation}
where $S$ is the set of simultaneously transmitting links. We will
sometimes use this version of the SINR constraint instead of Eqn. \ref{eq:sinr}.

For a set of points $P$, we will use $T(P)$ to denote a minimum spanning tree over the points in $P$.
We will simply use $T$ when $P$ is clear from the context. Naturally, $T$ contains undirected edges,
but when scheduling directed links, we need to orient $T$ in some way. When no ambiguity
arises, we will simply use $T$ to describe a particular oriented version of $T$.

\section{Related Work}
\label{sec:related}
Abstract problems capturing aspects of wireless networks have a long history, but the adoption of the SINR
model in theoretical analysis has been a comparatively recent phenomenon. The first rigorous worst case
results were achieved in the seminal work of Moscibroda and
Wattenhofer \cite{MoWa06} (which involved the problem studied in this paper). Ever since, numerous paper
have appeared on the SINR model. For a recent overview, see 
\cite{GoussevskaiaPW10}. Apart
from the connectivity, another fundamental problem is the capacity problem, where one wants to find
the maximum feasible subset of a given set of links. First rigorous results for the capacity problem 
were established in \cite{GHWW09}, followed
by a number of other results.
Kesselheim achieved a breakthrough recently by proving the first $O(1)$-approximation algorithm for capacity
with power control \cite{KesselheimSoda11}, whose techniques we adopt into our analysis. In this regard,
this work can be considered to bring the approaches to connectivity and capacity together. Other recent progresses
made include a $O(1)$-approximation capacity algorithm for oblivious powers \cite{SODA11}, and the study of topological properties of wireless communication maps \cite{stoc_topology11}.

\section{$O(\log n)$ connectivity in the SINR model}
\label{sec:conn}

The starting point of our analysis is a criteria for wireless capacity
recently developed by Kesselheim \cite{KesselheimSoda11}. Kesselheim
showed that any set of links for which this criteria holds (defined in
Eqn.~\ref{feasibilitycond} below) can be scheduled in a single slot,
and provided an efficient algorithm to do so. We shall 
call this algorithm \textbf{Schedule}, which is described in Section 3 of
\cite{KesselheimSoda11}. 
For reference, 
we also include the algorithm in Appendix \ref{app:a}.

Our approach is as follows. We show, via a related criteria, that
given any $T' \subseteq T(P)$, Eqn.~\ref{feasibilitycond} holds for
a constant fraction of the links in $T'$. Thus, a constant fraction of the
tree can be scheduled in a single step (and this holds
recursively). Naturally this process will end in $O(\log n)$
steps. Our analysis applies to any orientation of $T$. Thus to achieve
a strongly connected network, we simple schedule two trees. One is a
copy of $T$ oriented towards some arbitrary root, another one oriented
away from the same root. Thus any two nodes in the network can
communicate by first routing from the source to the root, and then
routing from the root to the destination.

Our goal is then to prove the following.
\begin{theorem}
Let $P$ be any set of points on the Euclidean plane. Let $T = T(P)$ be a minimum spanning tree on the points of $P$,
arbitrarily oriented. Then algorithm {\bf Connect} schedules $T$ in $O(\log n)$ slots.
\label{mainth1}
\end{theorem}

\begin{algorithm}                      
\caption{Connect(An arbitrarily oriented MST $T$ on point set $P$)}          
\label{alg1}                           
\begin{algorithmic}[1]                    
     \STATE $L \gets T$
     \WHILE{$L \neq \emptyset$}
     	\STATE Use Algorithm \textbf{Schedule} to find a feasible subset $L' \subseteq L$ \label{alg:findfeasibleset}
	\STATE $L \gets L \setminus L'$
     \ENDWHILE
\end{algorithmic}
\label{alg1fig}
\end{algorithm}

For two links $\ell, \ell'$, define $d(\ell, \ell') = \min\{d_{\ell \ell'}, d_{\ell' \ell}\}$. For links $\ell \leq \ell', \ell \neq \ell'$ define
the function $f_\ell(\ell') = \min\left\{1, \frac{\ell^{\alpha}}{d(\ell, \ell')^{\alpha}}\right\}$. Let $f_{\ell}(\ell) = 0$ and for $\ell > \ell'$
let $f_{\ell}(\ell') = 0$. The function $f_{\ell}(\ell')$ can be thought of as a measure of how badly the link $\ell$ might affect link $\ell'$ if
they were to transmit simultaneously.

We call a set of links $L$ \textbf{\amenable} if the following holds:
for any link $\ell = (s, r)$ ($\ell$ not necessarily a member of $L$), 
\begin{equation}
\sum_{\ell' \in L, \ell' \geq \ell} f_{\ell}(\ell') \leq \rho
\end{equation}
for some constant $\rho$ to be chosen later. The concept of an amenable set is closely related to the following
theorem due to Kesselheim (the connection is made explicit in Lemma \ref{lem:constantsizeapprox}).

\begin{theorem}[\cite{KesselheimSoda11}]
Assume $L'$ is a set of links such that for all $\ell' \in L'$,
\begin{equation}
\sum_{\ell \in L', \ell \leq \ell'} f_{\ell}(\ell') \leq \gamma
\label{feasibilitycond}
\end{equation}
for a constant $\gamma = \frac{1}{4 \cdot 3^{\alpha} \cdot (4 \beta + 2)}$.
Then, $L'$ is feasible and there exists a polynomial time algorithm to find a power assignment to schedule $L'$ in a single slot.

Moreover, for any given set $L$ assume $S$ is the largest feasible subset of $L$. Then \textbf{Schedule} finds a $L''$ of size $\Omega(S)$ for which Eqn.~\ref{feasibilitycond} holds.
\label{kessel}
\end{theorem}

\begin{lemma}
If a set $L$ of size $n$ is amenable, then there are $\Omega(n)$ links in $L$ that can be scheduled in a single slot.
\label{lem:constantsizeapprox}
\end{lemma}
\begin{proof}
Since $L$ is amenable, then by definition $\sum_{\ell \in L} \sum_{\ell' \geq \ell} f_{\ell}(\ell') \leq n \rho$. Rearranging,
we get $\sum_{\ell' \in L} \sum_{\ell \leq \ell'}  f_{\ell}(\ell') \leq n \rho$. By an averaging argument, there
must be a set $S$ of at least $n/2$ links for which 
\begin{equation}
\sum_{\ell \leq \ell'}  f_{\ell}(\ell') \leq 2 \rho \ .
\end{equation}
This is almost exactly Eqn.~\ref{feasibilitycond} except for the use
of a different constant. To achieve the correct 
constant, a simple sparsification suffices. Start an empty set.
Go through links in $S$ in increasing order of length, putting the link in the first set in which Eqn.~\ref{feasibilitycond} holds.
Start a new set if necessary. Clearly, no more that $\frac{2 \rho}{\gamma}$ sets will be necessary.

Thus, a set of size $\frac{n \cdot \gamma}{4 \rho}$ can be found for which Eqn.~\ref{feasibilitycond} holds.
\end{proof}

The most important step is to prove that $T$ is amenable.
\begin{lemma}
Let $T'' \subseteq T$ where $T = T(P)$ is a minimum spanning tree on the points in $P$. Then $T''$ is amenable.
\label{lem:approx}
\end{lemma}
\begin{proof}
Consider any link $\ell$ (not necessarily in $T$) and assume without of loss of generality that its length is $1$. 
We can do this because scaling all links to make $\ell$ of length $1$ does not change the values of the function 
$f_{\ell}(\ell')$.
To prove
amenability, we thus have to only consider links in $T''$ of length at least $1$. Let $T'$ be this set and let $P'$ be the points that are incident to at least one edge in $T'$.

First we claim,
\begin{lemma}
Any disc of radius $c_1 = 1/4$ contains at most 9 points from $P'$. 
\label{pointsincircle}
\end{lemma}
\begin{proof}
Let $D$ be a disc of radius $c_1$ and 
let $P_D$ be the set of points from $P'$ in $D$.  

We first observe that no two points $p_1, p_2 \in P_D$ have a common
neighbor in $T'$.
If $p_1$ and $p_2$ were neighbors in $T$ then a common neighbor would imply a cycle, while if they were non-neighbors,
replacing either of the edges to the common neighbor by the edge $(p_1, p_2)$ 
results in a cheaper spanning tree (since $d(p_1, p_2) \leq 2 c_1 < 1$).
Since each point in $P_D$ has a neighbor in $T'$, it holds that $|N(P_D)| \geq |P_D|$ (where $N(X) = N_{T'}(X) = \{p \in P' : \exists x \in X, (x,p) \in T'\}$ 
denotes the neighborhood of a point set $X$ in $T'$).

Let $c$ be the center of $D$, and consider any pair of points $a, b \in N(P_D)$. 
We aim to show that the angle $\angle a c b > \pi/5$, which implies the lemma.
Let $p_a$ ($p_b$) be the unique neighbor of $a$ ($b$) in $T$.
We observe first that the unique path in $T$ between $p_a$
and $p_b$ goes through neither $a$ nor $b$, since if it did, say
through $a$, then replacing $(p_a, a)$ by $(p_a, p_b)$ results in a 
smaller tree.  

Consider now the triangle $\triangle a b c$.  Let $\alpha = \angle c a b$, 
$\beta = \angle a b c$, $\gamma = \angle b c a$, and denote
$|p_1 p_2| = d(p_1, p_2)$, for points $p_1$ and $p_2$.
Note that $|ab| \ge d(a, p_a) \ge 1$, as
we could otherwise delete the edge $ap_a$ and add $ab$ to get a
better tree. Similarly, $|ab|\ge d(b, p_b) \ge 1$.
From the triangular inequality, our observations above, and the
fact that $|ab| \ge 1$, we have that
\[ |ac| \le d(a,p_a) + d(p_a,c) \le |ab| + c_1 \le |ab| (1 + c_1)\ , \]
and similarly $|bc| \le {|ab|} (1 + c_1)$.
By the sine law, $\sin \alpha / \sin \gamma = |bc|/|ab| \le 1 + c_1$,
and $\sin \beta / \sin \gamma = |ac|/|ab| \le 1 + c_1$.
For $c_1 = 1/4$, this implies that since $2\pi = \alpha + \beta + \gamma 
\le \gamma + 2 \arcsin(5/4 \cdot \sin \gamma)$, computation shows that $\gamma > \pi/5$ as claimed.
\end{proof}

Now for $\ell = (s, r)$, $\sum_{\ell' = (s',r') \in T', \ell' \geq \ell} f_{\ell}(\ell')
  \leq \sum_{p \in P'}  \left(\min\left\{1, \frac{1}{d(p,s)^{\alpha}}\right\} 
   + \min\left\{1, \frac{1}{d(p,r)^{\alpha}}\right\}\right).$
Thus it suffices to upper bound $\sum_{p \in P'} \min\left\{1,
\frac{1}{d(p, x)^{\alpha}}\right\}$ for any arbitrary point $x$ by a constant
to get the required bound.

Now take concentric circles $C_0, C_1 \ldots$ around $x$ such that the $t^{th}$ 
circle has radius $t+1$.
The proof of the following fact
can be found in Appendix B. 
\begin{lemma}
$C_0$ can be covered by $O(1)$ circles of radius $c_1$
 (where $c_1$ is the constant from Lemma \ref{pointsincircle}). 
The annulus $C_t \setminus C_{t-1}$ can be covered by $O(t)$ circles of radius $c_1$, for $t \geq 1$.
\label{ddim}
\end{lemma}
Thus by Lemma \ref{pointsincircle}, there are at most $O(t)$ points from $P'$ in $C_t \setminus C_{t-1}$. Similarly,
$C_0$ can be covered by $O(1)$ circles of radius $c_1$ (see Lemma \ref{ddim}) and thus contains $O(1)$ points from $P'$.

The distance to $x$ from any point in $C_t \setminus C_{t-1}$ is at least $t$. Then,
\begin{align*}
\lefteqn{\sum_{p \in P'}  \min\left\{1, \frac{1}{d(p, x)^{\alpha}}\right\}}  \\
& = |p \in P' \cap C_0| \cdot 1 + \sum_{t \geq 1} \sum_{p \in P' \cap (C_t \setminus C_{t-1})} \frac{1}{d(p, x)^{\alpha}} \\
& \leq O(1) + O\left(\sum_{t \geq 1} t \frac{1}{t^{\alpha}}\right) = O(1)\ ,
\end{align*}
%
for $\alpha > 2$. This completes the proof of Lemma \ref{lem:approx},
assuming that $\rho$ is at least twice the implicit constant in
the bound above.
\end{proof}

Theorem \ref{mainth1} now follows easily. By Lemma \ref{lem:approx}, the remaining set of links at each step
of the algorithm is amenable. Thus, by Lemma
\ref{lem:constantsizeapprox}, a constant factor of these links are feasible, and by Theorem \ref{kessel}, a constant
factor of those will be scheduled by {\bf Connect}. Clearly, this process terminates in $O(\log n)$ steps.

\smallskip
\noindent
\textbf{Remark.} We note that the assumption $\alpha > 2$ is necessary.
Indeed, suppose points are placed at all integer coordinates within a
large circle, so all links will be of at least unit length.
Then, when $\alpha \le 2$, it can be shown with standard methods that 
there is no feasible subset of links of size larger than $\Omega(n/\log n)$.

\subsection{A lower bound on our approach}
It is easy to construct an example where a link in $T(P)$ violates
Eqn.~\ref{feasibilitycond} (below, the set $G_1$ provides an example of that). However,
this still leaves open the possibility that the spanning tree can be partitioned into a small number
 of subsets (for example, a constant number of subsets) such that Eqn.~\ref{feasibilitycond} holds
for each of them, thus improving upon the $O(\log n)$ result. In the following theorem, we show that one
cannot partition all the points into a constant number of subsets. This, naturally, is not a lower bound
on the connectivity problem, just on our particular approach.

\begin{theorem}
For any number $c$, there exists a set of points on the line such that the minimum spanning tree $T$ cannot be 
partitioned into $\leq c$ sets $S_1, S_2 \ldots S_c$ such that Eqn.~\ref{feasibilitycond} holds for each $S_i$.
\end{theorem}
\begin{proof}
For $t \geq 1$, we will recursively construct gadgets $G_t$ such that a spanning tree on $G_t$ cannot be partitioned into $t$
sets for which Eqn.~\ref{feasibilitycond} holds. 

Since we are considering points on a line, the minimum spanning tree is simply
the edges connecting each point to its immediate neighbors to the right and left. Our theorem holds
for any orientation of the links.

A gadget $G$ is simply a set of points located on a line, with an implicit ordering from the left to the right. We will often use $G$ to mean a translated copy
of $G$ as well, which will be clear from the context.
For two gadgets $F$ and $G$, we will use $F \oplus G$
to denote the joining of the two gadgets, which is a new gadget with $|F| + |G| -1$ points. The first (starting from the left) $|F|$ points
are a copy of $F$, and the last $|G|$ points are a translated copy of $G$. In other words, the $|F|^{th}$ point is both
the ending point of the gadget $F$ and the starting point for the copy of gadget $G$. For any collection
of points (or gadget) $G$, let $L(G)$ be the diameter of $G$.

For a gadget $G$, we use $G(b)$ to mean a copy of the gadget scaled by a factor of $b$. For example,
if $G = \{-10, 0, 1, 2.5\}$, then $G(10) = \{-100, 0, 10, 25\}$.

\begin{figure*}
\begin{center}
\includegraphics[height=2.5cm]{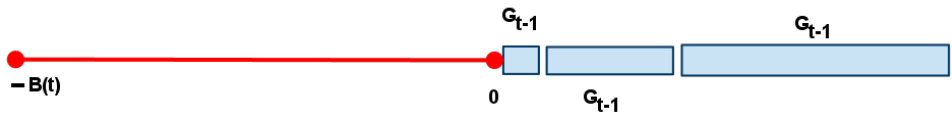}
\caption{Construction of $G_t$ from copies of $G_{t-1}$. We join a number of copies of $G_{t-1}$, each succeeding copy
scaled must larger than all the copies before it combined. All this is preceded by a huge new link between points $-B(t)$ and $0$ ($B(t)$ is very large)}
\label{gtconsfig}
\end{center}
\end{figure*}

We are ready to describe our construction. $G_1$ contains the points $\{-28, 0, 2, 6, 14\}$.
For a gadget $G$, define $\rho(G) = \min_{\ell \in T(G)} \frac{\ell^{\alpha}}{{\hat d_{\ell}(G)}^{\alpha}}$ where
$\hat d_{\ell}(G)$ is the maximum distance from either end point of $\ell$ to the left most point in $G$. 
It is easy to verify that $\rho(G) \leq 1$.
Now $G_t$ is constructed by joining copies of $G_{t-1}$ (each copy scaled much larger than preceding copies) and preceding all of it with a new point, such that
the distance between the new point and the beginning of the copies of $G_{t-1}$ is humongous. This
is informally depicted in Fig.~\ref{gtconsfig}. 

To formally define $G_t$, we first define $G'_t = \oplus_{1 \leq i \leq I(t)} G_{t -1}(h(i))$. The value of $I(t)$ is set to 
$\lceil \frac{8^{\alpha}  \gamma}{ \cdot \rho(G_{t -1})} \rceil$ where $\gamma$ is the constant from Eqn. \ref{feasibilitycond}. The scaling factors are defined as such: $h(1) = 1$, $h(j)$ (for $j > 1$) is chosen such that
$\min_{\ell \in T(G_{t -1}(h(j)))} \ell = 2 L(\oplus_{1 \leq i \leq j-1} G_{t -1}(h(i)))$. Let $G'_0$ be $G'_t$ translated so that its left most point is at location $0$. Define the number $B(t) = 4 L(G'_t)$. Then $G_t = \{-B(t)\} \cup G'_0$, i.e., a single point at location $-B(t)$, followed by the gadget $G'_0$. 

Define the partition number $P(G)$ as the minimum number of partitions of a link set required so that Eqn.~\ref{feasibilitycond}
holds for each set. We claim:
\begin{claim}
$P(G_1) \geq 2$ and $P(G_t) \geq 1 + P(G_{t-1})$ for all $t \geq 1$.
\end{claim}
Combined, these two claims clearly prove the theorem. The claim about $G_1$ is easy to verify by direct computation.
Let us prove the inductive step.
For $G_t$ consider the left-most link $\ell_g$ in $T = T(G_t)$. This is of course a link of length $B(t)$ which is by construction
the unique largest link in $G_t$ (since $B(t) = 4 L(G'_t)$). 
Thus $\{\ell \in T : \ell < \ell_g\} = T \setminus \{\ell_g\}$. 
We claim that
$\frac{\ell^{\alpha}}{d(s, r_g)^{\alpha}} +  \frac{\ell^{\alpha}}{d(s_g, r)^{\alpha}} \geq \frac{\rho(G(t-1))}{2^{\alpha}}$ for all $\ell \in T \setminus \{\ell_g\}$. To see this, consider any $\ell \neq \ell_g \in T$. 
Now $\ell$ is of course part of the $j^{th}$ copy of $G_{t-1}$ for some $j \geq 1$. 
Now let $d' = \min\{d(s, r_g), d(s_g, r)\}$.
We can  observe that 
$d' \leq \hat d_{\ell}(G_{t-1}(h(j)) + L(\oplus_{1 \leq i \leq j-1} G_{t -1}(h(i)))$. By construction, 
$\hat d_{\ell}(G_{t-1}(h(j)) \geq \ell > L(\oplus_{1 \leq i \leq j-1} G_{t -1}(h(i)))$ and thus 
$\frac{\ell^{\alpha}}{d(s, r_g)^{\alpha}} +  \frac{\ell^{\alpha}}{d(s_g, r)^{\alpha}} \geq  \frac{\ell^{\alpha}}{d'^{\alpha}} \geq \frac{\rho(G_{t-1})}{2^{\alpha}}$ (by definition of $\rho(G_{t-1}$)) proving the claim. 

Given this and noticing the value of $I(t)$ chosen in the construction of
$G'_t$, it is clear that there must be some $1 \leq j \leq I(t)$ such that no link in $T(G_{t-1}(h(j)))$
is in the same set as $\ell_g$. This completes the proof of the claim and the theorem.
\end{proof}

\section{Bi-directionality}
\label{sec:bidi}
We have worked with the uni-directional model of wireless communication so far, where the links are directed from
a sender to a receiver. The bi-directional model, in contrast, has
two-way half-duplex communication
between the nodes of a link in the same slot. The advantage of this model is that it simplifies one-hop communication protocols. Two-way communication in a single slot without worrying about mutual interference can be achieved in practice in more than one way,
and we simply take that as given. The difficulty arises in that interferences from other links are now potentially
much larger, since we have to take into account both directions of each link.

We can model the bi-directional case as follows: $L$ contains
$n$ pairs $\ell = \{n_1, n_2\}$. These two points implicitly define two unidirectional links $\ell_1 = (n_1, n_2)$ and $\ell_2 = (n_2, n_1)$. Each pair can be associated with two power levels $P_{\ell_1}$ and $P_{\ell_1}$, to be used by 
$\ell_1$ and $\ell_2$ respectively. We consider a set of pairs $L'$ feasible if for all $\ell \in L'$,
\[  \frac{P_{\ell_i}/(\ell_i)^\alpha}{\sum_{\ell' \in L' \setminus  \{\ell\}}   \sum_{k = 1}^2 P_{\ell'_k}/(d_{\ell'_k \ell_i})^\alpha + N} \ge \beta \quad \mathrm{for } i \in \{1, 2\}, \]
or equivalently,
\[ \sum_{\ell' \in L', \ell' \neq \ell} (a_{\ell'_1}(\ell_i) +
a_{\ell'_2}(\ell_i)) \leq 1 \quad \mathrm{for } i \in \{1, 2\}, \]
where affectances are as defined for the unidirectional case.

We differentiate two versions of the bi-directional model. In the \emph{symmetric} model, we insist that $P_{\ell_1} = P_{\ell_2}$
for each $\ell$. With this restriction, the model is essentially equivalent to the one
introduced in \cite{FKRV09}.
Without such a restriction,
we call it the \emph{asymmetric} model, which was briefly mentioned in \cite{KesselheimSoda11}.

First we show,
\begin{theorem}
There is an instance that 
requires $\Omega(n)$ slots for connectivity in the symmetric bi-directional model.
\label{bidirlb}
\end{theorem}
\begin{proof}
Consider the pointset $x_0, \ldots, x_{n-1}$ given by
$x_0 = 0, x_1 = 1$ and for $i \ge 2$, $x_{i} = 2 x_{i-1}^2$.
Observe that $x_{m} - x_{m-1} > x_{m-1}^2$.
Consider any two pairs $\ell = \{x_i, x_j\}$ and $\ell'=\{x_k, x_m\}$,
and assume without loss of generality that $m \geq \max(i,j,k)$.
Thus, there must be indices $a, b, c, d \in \{1, 2\}$
such that $d_{\ell_a \ell'_b} \leq x_{m-1}$ and $d_{\ell'_c \ell_d} \leq x_{m-1}$.

On the other hand $\ell'_b \ge x_{m} - x_{m-1}$ and $\ell_d \geq 1$.

Then,
 \begin{align*}
a_{\ell_a}(\ell'_b) \cdot a_{\ell'_c}(\ell_d) = \frac{P_{\ell_a}/d_{\ell_a \ell'_b}^\alpha}{P_{\ell'_b}/(\ell'_b)^\alpha}
\cdot \frac{P_{\ell'_c}/d_{\ell'_c \ell_d}^\alpha}{P_{\ell_d}/(\ell_d)^\alpha} \\
 = \left( \frac{\ell_d \ell'_b}{d_{\ell'_c \ell_d} d_{\ell_a \ell'_b}} \right)^\alpha
 \ge  \left( \frac{x_m - x_{m-1}}{x^2_{m-1}}\right)^{\alpha}
 > 1\ .
 \end{align*}
Thus, any pair of links must be scheduled in different slots.
\end{proof}

In surprising contrast to the above strong lower bound in the symmetric model,
\begin{theorem}
In the asymmetric bi-directional model, any set of $n$ points
can be strongly connected in $O(\log n)$ slots.
\end{theorem}

The argument in Section \ref{sec:conn} is as follows. First, we show that $T' \subseteq T$ is amenable. Then
we find a large subset for which Eqn \ref{feasibilitycond} holds, and finally, we schedule it in one slot.
The main difference in the bi-directional case is that we have to choose pairs in a feasible set, i.e., for any
pair $\ell$ that we want to connect, we have to include $\ell_1$ and $\ell_2$ in the \emph{same} slot.
Note that since, $\ell_1$ and $\ell_2$ have no effect on each other, we can define $f_{\ell_i}(\ell_k) = 0$ for
$i, k \in \{1, 2\}$. The new definition of amenability is thus,

\begin{equation}
\sum_{\ell' \in L, \ell' \geq \ell} \sum_{j = 1}^2 \sum_{k = 1}^2 f_{\ell_j}(\ell'_k) \leq \rho
\end{equation}
for a constant $\rho$.

We first need to verify that Lemma \ref{lem:approx} still holds with the new definition. This happens to be easy. Indeed, the proof of Lemma \ref{lem:approx} does not use the dichotomy between sender and receiver, and thus automatically holds (up to a factor of 4). It is easy to see that the argument in Lemma \ref{lem:constantsizeapprox} continues to hold with minor differences.

Finally, we need to show that Thm. \ref{kessel} still holds with the new definition, i.e., the algorithm \textbf{Schedule} can still successfully find
and schedule the link set thus selected. The algorithm \textbf{Schedule} is robust in relation this, as \cite{KesselheimSoda11} points out. 

More specifically, to show that \textbf{Schedule} works for the bi-directional variant, we need the following version of 
Thm. \ref{kessel}:

\begin{proposition}
Assume $L'$ is a set of pairs such that for all $\ell' \in L'$,
\begin{equation}
\sum_{\ell \in L', \ell \leq \ell'} \sum_{j=1}^2 \sum_{k=1}^2 f_{\ell_j}(\ell'_k) \leq \gamma
\label{feasibilitycond2}
\end{equation}
for a constant $\gamma = \frac{1}{4 \cdot 3^{\alpha} \cdot (4 \beta + 2)}$.
Then, $L'$ is feasible and there exists a polynomial time algorithm to find a power assignment to schedule $L'$ in a single slot.

Moreover, for any given set $L$ assume $S$ is the largest feasible subset of $L$. Then \textbf{Schedule} finds a $L''$ of size $\Omega(S)$ for which Eqn.~\ref{feasibilitycond2} holds.
\label{kessel2}
\end{proposition}

The last part about finding a large feasible subset is an implication of arguments like Lemmas \ref{lem:approx} and \ref{lem:constantsizeapprox}, which we already verified to be sound in the new regime.

For the first part of the algorithm, Eqn. \ref{feasibilitycond2} is
identical to Eqn. \ref{feasibilitycond} if we assume the link set to
be $L' = \cup_{\ell} \{\ell_1, \ell_2\}$ except one caveat. For a
given $\ell$, Eqn. \ref{feasibilitycond} includes the term
$f_{\ell_k}(\ell_j)$ for the two links of the same pair, and
Eqn. \ref{feasibilitycond2} doesn't (or rather $f_{\ell_k}(\ell_j)$ is
set to zero). However, this is not a problem, since we assume that
$\ell_1$ and $\ell_2$ do not interfere with each other. In relation to
all other pairs, Eqn. \ref{feasibilitycond2} is identical to
Eqn. \ref{feasibilitycond} and thus the argument is identical.

\section{Oblivious Power Assignments}
\label{sec:oblivious}

We examine here the complexity of connectivity when using simple power
assignments. A power assignment is said
to be \emph{oblivious} if it depends only on the length of the link.
We show that any reasonable oblivious power assignment is ineffective
in that it requires $\Omega(n)$ slots to connect some instance of $n$ points.
On the other hand, we also find that if the diversity of the edge lengths in the MST is
small, then they can be quite effective. 


Moscibroda and Wattenhofer \cite{MoWa06} showed that for both uniform power (all links use the same power)
and linear power ($P_{\ell} = \ell^{\alpha}$), there are pointsets for which connectivity requires
$\Omega(n)$ slots.  It is easy to verify that their construction
applies also to functions that grow slower than uniform (i.e., are decreasing)
or faster than linear.
We address here essentially all other reasonable
oblivious assignments, namely those that are monotone increasing but
grow slower than linear.  

We call a power function $p$ 
\emph{smooth} if $p(x) \geq x$ for all $x$, $p(x) \le p(y)$ when $x \le y$, and $p(x) \le x^\alpha$,
and $g:\mathbb{R}^+ \rightarrow
\mathbb{R}^+$  defined by $g(x) = \frac{1}{2}\min(p(x),
x^{\alpha}/p(x))$ is monotone
increasing and $\omega(1)$. This is true for mean power ($p(x) = x^{\alpha/2}$) and many similar power assignments (such as the
one used in \cite{DBLP:conf/wdag/KowalskiR10}, $p(x) = x^{(\alpha+2)/2}$).

\begin{lemma}
Let $Y = \{y_1, y_2 \ldots y_n\}$ be a set of points on the line such that $y_1 < y_2 < \ldots < y_n$, 
the minimum distance between any pair of points is $1$, and 
$g(y_t - y_{t-1}) \geq  (y_{t-1} - y_1)^\alpha$, for each $t=3, 4, \ldots, n$. 
Then, no two links between points in $Y$ can be scheduled simultaneously using power assignment $p$.
\label{lem:incr1}
\end{lemma}

\begin{proof}
Consider two links $\ell_t=(y_t,y_k)$ and $\ell_1=(y_i, y_j)$, where 
without loss of generality $t \ge \max(k,i,j)$ and 
$y_i$ is the sender of $\ell_1$.  
We may assume that $j \neq t$ and $i \neq t$, since a point
cannot be involved in two transmissions simultaneously, if the signal
requirement $\beta >
1$.  The power is $P_t = p(y_t - y_k)$ on link $\ell_t$ and $P_1 = p(|y_i
- y_j|)$ on link $\ell_1$.

First, consider the case where $y_t$ is the receiver of $\ell_t = (y_k,
y_t)$, $k < t$. 
The affectance of $\ell_t$ on $\ell_1$ is
\begin{align*}
 a_{t}(1) & = \frac{p(y_t - y_{k})}{|y_k - y_j|^{\alpha}}
   \frac{|y_i - y_j|^{\alpha}}{ p(|y_i - y_j|)}
 \overset{1}{\geq} \frac{p(y_t - y_{t-1})}{|y_k - y_j|^{\alpha}} \\
& \overset{2}{\geq} \frac{p(y_t - y_{t-1})}{g(y_t - y_{t-1})} 
 \overset{3}{>} 1\ .
\end{align*}
Explanations:
\begin{enumerate}
\item By sublinearity, $p(|y_i-y_j|) \le |y_i - y_j|^\alpha$, 
  and by monotonicity, $p(y_t - y_k) \ge p(y_t - y_{t-1})$.
\item Because $|y_k - y_j|^\alpha \le (y_{t-1} - y_1)^\alpha \le g(y_t - y_{t-1})$.
\item Since $g(x) < p(x)$.
\end{enumerate}
Thus these two links cannot be scheduled together.

Second, consider the case where $y_t$ is the sender of $\ell_t = (y_t,
y_k)$, $k < t$.  Let $z$ denote $y_t - y_k$.
The affectance of $\ell_1$ on $\ell_t$ is now
\begin{align*}
a_{1}(t) & = \frac{p(|y_i - y_j|)}{|y_i - y_k|^{\alpha}}
     \frac{z^{\alpha}}{p(z)}
\overset{1}{\geq}  \frac{z^{\alpha}}{|y_i - y_k|^\alpha p(z)} \\
& \overset{2}{\geq} \frac{z^{\alpha}}{g(z) p(z)}
\overset{3}{>} 1 \ .
\end{align*}
Explanations:
\begin{enumerate}
\item Because $p(|y_i - y_j|) \ge 1$, since $|y_i - y_j| \ge 1$.
\item Since $|y_i - y_k|^\alpha \le (y_{t-1} - y_1)^\alpha \le g(z)$.
\item Since $g(z) < z^\alpha/p(z)$, by assumption.   
\end{enumerate}
\vspace{-2em}
\end{proof}

We shall argue the lower bound for a more general class of structures
(similar to \cite{MoWa06}).
We say that a structure (set of links) on a pointset has property
$\phi_{min}$ if each point is either a sender or receiver on at least
one link.

Since $g$ is monotone increasing and eventually infinite, it has an
inverse $g^{-1}$.
We construct $n$ points $x_1, x_2 \ldots x_n$ on the line defined by
$x_1 = 0$, $x_2 = 2$ and $x_i = x_{i-1} + g^{-1}(2(x_{i-1})^\alpha)$.
The following result is now immediate from Lemma \ref{lem:incr1}.

\begin{theorem}
For any structure with property $\phi_{min}$
and any smooth oblivious power assignment,
there is an instance that requires $n/2$ slots.
\end{theorem}

In spite of this highly negative statement, we do find that oblivious
power assignments are quite effective given some natural assumptions
about edge length distributions. 

\subsubsection*{Upper bounds}

Let $T$ be a MST of the given pointset.  Let $\Delta$ denote the ratio
between longest to shortest edge length in $T$.  Assume, by scaling,
that $\ell \geq 1$ for all $\ell \in L$. Let $g(L) = |\{m: 
\exists \ell \in L, \lceil \lg \ell \rceil = m \}|$ denote the
\emph{length diversity} of the link set $L$, or the number of length
groups. Note that $g(L) \le \log \Delta$.

\begin{theorem}
Any pointset can be strongly connected in $O(g(L))$
slots using uniform (or linear) power assignment.
This is achieved on an orientation of the minimum spanning tree.
\label{logdelta}
\end{theorem}

\begin{proof}
Divide the links of the spanning tree into at most $g(L)$ length
classes, where links in the same class differ in length by a factor at
most 2. Consider one such color class $S$.
Let $P$ be the endpoints of links in $S$ and let $d$ be length of
the shortest link.
Consider an endpoint $x$ of a link in $S$.
By Lemma \ref{pointsincircle}, at most 9 points from $P$ are within a
distance $d/4$ from any point. Note that any radius-$r$ circle can be
covered with at most $2(r/s)^2$ radius-$s$ circles.
Thus, for any $t \ge 1$, there are at most $C = 4 \cdot 2(4t)^2$ points
from $P$ within a distance $td$ from $x$.
The links in $S$ can then colored with $C$ colors so that senders of any pair of
links are of distance at least $td$. 
If $t = 4(\alpha 4^2 \tau(\alpha-1))^{1/\alpha}$, where $\tau$ is the
Riemann function and $\alpha > 2$, then it follows from Lemma 3.1 of
\cite{us:esa09full} that each colorset forms a feasible set using
uniform power.
The total number of slots used is then $t g(L)$.
\end{proof}

This bound improves on a bound of $O(g(L)\log n)$ given by Moscibroda
and Wattenhofer \cite{MoWa06}. The construction in \cite{MoWa06} shows
also that the bound is best possible for uniform and linear power.

Exponentially weaker dependence on $\Delta$ can be achieved by using
mean power (the power is set proportional to the length to the power of $\alpha/2$).

\begin{theorem}
Any pointset can be strongly connected in $O(\log n(\log\log \Delta +
\log n))$ slots using mean power assignment.
\label{loglogdelta}
\end{theorem}

\begin{proof}
The capacity of a linkset is the maximum number of links that can be
scheduled simultaneously.  Our main result is that any orientation of
the MST $T$ yields a directed linkset with linear capacity: $\Omega(n)$ 
links can be scheduled in a single slot. 
%
A recent result \cite{SODA11} shows that for any linkset, the
optimal capacity with power control differs from optimal capacity with
mean power by a factor of $O(\log\log \Delta + \log n)$.
Further, a constant approximation algorithm for mean power capacity is
given in \cite{SODA11}. That algorithm then schedules $\Omega(n/(\log
\log \Delta + \log n))$ links from $T$ in a single slot. 
In $O(\log n (\log\log \Delta + \log n))$ slots it will then have
scheduled all of $T$.
\end{proof}

Note that the construction of Lemma \ref{lem:incr1} yields a lower
bound of $\Omega(\lg\lg \Delta)$ for mean power.

\section{Extensions to other connectivity problems}
\label{sec:design}

\subsection{Minimum-latency aggregation scheduling}

Recall the problem definition.  An \emph{in-arborescence} $T$ is a
directed rooted tree that has a path from every node to the root. An
edge $e$ in $T$ is said to be a \emph{descendant} of edge $e'$ if
there is a directed path starting with $e$ that includes $e'$.
Given a set of $n$ points $P$ on the plane, 
the MLAS problem is to
find $t$ ordered disjoint linksets $S_1, S_2, \ldots S_t$  such that
each $S_i$ is feasible, the links in $T = \cup_i S_i$ form a spanning in-arborescence $T$, and whenever $e \in S_i$ is a descendant of $e' \in S_j$ then $i < j$. Let us call this last condition the 
\emph{ordering requirement}.


Consider the following iterative algorithm. 
Let $P_1 = P$.
In step $i$ the algorithm finds a feasible linkset $S_i$ on $P_i$ and derives a new pointset $P_{i+1}$, repeating the process until $P_{i+1}$ contains only a single node.
Given $P_i$, we form the nearest-neighbor forest $F_i$, where each node $p \in P_i$ provides the link $(p, p')$ to its nearest point $p'$; whenever links
whenever $F_i$ contains a pair $(p, p')$ and $(p',p)$, we remove one of the two links. This forest $F_i$ is a subset of some minimum spanning tree of $P_i$, and therefore it is amenable by Lemma \ref{lem:approx}. 
Thus, we can find a feasible set $S_i \subseteq F_i$ with $|S_i| = \Omega(|P_i|)$
using \textbf{Schedule}.
This set $S_i$ is necessarily a (partial) matching on $P_i$.
We form $P_{i+1}$ by removing from $P_i$ the tails of all the links in $S_i$.

We first show that this algorithm uses $O(\log n)$ steps, which follows immediately from the following Lemma.

\begin{lemma}
$|P_{i+1}| \leq c_3 |P_{i}|$, for some $c_3 < 1$.
\end{lemma}
\begin{proof}
The forest $F_i$ contains at least $|P_i|/2$ edges.
By Theorem \ref{kessel}, \textbf{Schedule} finds a feasible matching $S_i$ of size at least $c_4 |F_i|$, for some $c_4 > 0$.
Then, $|P_{i+1}| = |P_i| - |S_i| \le (1 - c_4)|P_i$.
\end{proof}

We also need to show that the resulting link set forms an in-arborescence and that it satisfies the ordering requirement.
Both of these are easily verified.

Also, it can be easily verified that any aggregation tree satisfying the ordering requirement requires a schedule of length at least $\lg n$. 

Thus we get the following result.
\begin{theorem}
Given any set of $n$ points on the plane a aggregation tree can be formed with $O(\log n)$ latency, and this is optimal.
\end{theorem}

\subsection{Biconnectivity and $k$-edge connectivity}

We can use our basic connectivity method to achieve additional network
design criteria. 
As a warmup, we first show how to achieve biconnectivity at minimal extra cost.
A graph is biconnected if there are at least two vertex-disjoint paths between
any pair of vertices.
\begin{theorem}
Let $P$ be any set of points on the Euclidean plane. Then $P$ can be
strongly biconnected in $O(\log n)$ slots.
\label{mainth3}
\end{theorem}

To see this, take the minimum spanning tree $T$ used for
Thm.~\ref{mainth1}. Let $X$ be the set of degree-1 nodes in $T$ and
form a minimum spanning tree $T'$ of $X$. Apply the algorithm \textbf{Connect}
to the union of $T$
and $T'$, directed in both ways. Between any pair of nodes is a path in
$T$, all of whose internal nodes are in $P \setminus X$, and a path in $T'$,
with all its internal nodes in $X$.

\smallskip

A directed graph is $k$-edge strongly connected if the graph stays strongly connected after the removal of less than $k$-edges.
Here we prove:
\begin{theorem}
Let $P$ be any set of points on the Euclidean plane. Then $P$ can be $k$-edge strongly connected in $O(k^4 \log n)$
slots.
\label{mainth2}
\end{theorem}
\begin{proof}{[Outline]}
The algorithm is as follows. We repeatedly compute $k$ spanning trees
$T_0, T_1 \ldots T_{k}$. Here, $T_0$ is a minimum spanning tree, and
for $i \geq 1$, $T_i$ is a minimum spanning tree that does not use any
edge from $\cup_{j < i} T_j$. Once we schedule these trees in two
orientations, the resultant structure is clearly $k$-edge strongly
connected.

We then claim that each $T_i$ can be scheduled in $O(i^3 \log n)$
slots from which the theorem follows. Proving that $T_i$ can be
scheduled in $O(i^3 \log n)$ boils down to proving a version of Lemma
\ref{pointsincircle} for $T_i$, 
given below.
The rest follows in a routine fashion.
\end{proof}


\begin{lemma}
Any disc of radius $c_1 = 1/4$ contains at most $O(i^3)$ points from $P'$, where $P'$ is the set of points incident
to a link of length at least $1$ in $T_i$. 
\label{pointsincircle-kmst}
\end{lemma}
\begin{proof}
Consider $T_i$ for $i \geq 1$ (we already have the bound for $i = 0$). As before, let $N(P_D)$
be the set of neighbors of $P_D$ in $T_i$.

Define $G = \cup_{j < i} T_j$.

\begin{lemma}
Let $a, b \in N(P_D)$ with the following property:
There exist $p_1, p_2 \in P_D$ such that $(a,p_1), (b,p_2) \in T_i$,
and $(a,b), (p_1,p_2) \not\in G$.
Then $\angle a c b > \pi /5$.
\label{largeangle2}
\end{lemma}
The proof of this claim  is essentially identical to the same argument
in Lemma \ref{pointsincircle}.  The fact that $(a, b) \not \in G$ and
$(p_1, p_2) \not \in G$ simply mean that the links $(p_1, p_2)$ and
$(a, b)$ can be used in the argument as they are not ruled out by
being included in an earlier tree.

Assume from now on that $|P_D| \ge c_2 i^3$ for $c_2 = 72 \cdot 10$.
We shall show that there exists then a set $B \subseteq N(P_D)$ of 10 points all of
whose pairs satisfy the conditions of Lemma \ref{largeangle2}, which
leads to a contradiction. 
Let $G[X]$ denote the subgraph of $G$ induced by pointset $X$.

We first argue that $|N(P_D)| \geq \frac{1}{2i}|P_D|$. More strongly,
we claim that no point in $N(P_D)$ has more than $2 i$ neighbors in
$P_D$. Suppose point $p_1 \in N(P_D)$ has a set $X$ of $c > 2i$
neighbors in $P_D$.  Since $G$ is a union of $i-1$ trees, $G[X]$ contains
at most $(c -1)(i-1)$ edges, which is strictly smaller than
$\frac{c (c-1)}{2}$, as $c > 2 i$. Thus there is a pair $x_1, x_2 \in
X$ that is non-adjacent each of the previous trees, in which case
we can argue as in Lemma \ref{pointsincircle} and claim that we can
delete $(p_1, x_1)$ and add $(x_1, x_2)$ to get a better tree.


The following is a general claim about points in relation to spanning trees.
\begin{claim}
For any set $Y$ of points, $G[Y]$ contains
an independent set of size $\frac{|Y|}{2 i -1}$ in $G$.
\label{largeis}
\end{claim}
\begin{proof}
Since $G$ is a union of $i-1$ trees, the average degree of any induced
subgraph is less than $2(i-1)$. The claim then follows from Tur\'an bound.
\end{proof}

Now, by Observation \ref{largeis},
there is an independent set $Y_c \subseteq N(P_D)$ in $G$ of size at least
$\frac{1}{2i} |N(P_D)| \ge \frac{1}{(2i)^2} |P_D| \geq \frac{c_2}{4} i$.

If some ten points in $Y_c$ share a common neighbor in $P_D$, then we
are done. Otherwise, there is a subset $Y'$ of $Y_c$ of size at least 
$\frac{c_2}{4 \cdot 9} i$ such that no two share the same neighbor in $P_D$. 
Let, $Z \subseteq P_D$ be the neighbors of $Y'$ in $P_D$. By Observation \ref{largeis}, 
we can find a subset $Z' \subseteq Z$ of
size at least $\frac{c_2}{9 \cdot 8}$ which is independent in $G$.
Since no two points in $Y_c$ share neighbors in $Z'$, $|N(Z') \cap Y_c|
\geq \frac{c_2}{72}$. Setting $c_2 = 72 \times
10$, we find that $B = N(Z') \cap Y_c$ contains at least 10 points all of whose
pairs satisfy the conditions of Lemma \ref{largeangle2}, which is a contradiction.
Hence, $|P_D| \le c_2 i^3$.
\end{proof}

\section{Conclusion}

We have shown that there the links of a minimum spanning tree of any pointset can be be scheduled in $O(\log n)$ slots in the SINR model. An open question is whether this is optimal; we conjecture that it is.
Another direction would be to derive effective distributed algorithms.

\bibliographystyle{plain}
\bibliography{references}		

\appendix

\section{The algorithm Schedule}
\label{app:a}
We include, as a reference, the algorithm \textbf{Schedule} due to Kesselheim \cite{KesselheimSoda11}.

\begin{algorithm}                      
\caption{Schedule (Set $L$ of $n$ links)}          
\label{alg2}                           
\begin{algorithmic}[1]                    
     \STATE Sort links in increasing order of length $\ell_1 \leq \ell_2 \ldots \leq \ell_n$, breaking ties arbitrarily
     \STATE $S \gets \emptyset$
     \FOR{$i = 1$ to $n$} 
     \IF{$\sum_{j < i} f_{\ell_j}(\ell_i) \leq \gamma$} 
     \STATE $S \gets S \cup \{\ell_i\}$
      \ENDIF
     \ENDFOR

     \STATE Now schedule $S$ by finding power assignment for all links in $S$:
     \STATE $P_{\ell_n} = 1$
     \FOR{$i = n-1$ to $1$} 
     \STATE $P_{\ell_i} = 4 \beta \cdot \sum_{j > i} \frac{P_{\ell_j} \ell_i^{\alpha}}{d(s_j, r_i)^{\alpha}}$ where $s_j$ is the sender of $\ell_j$ and $r_i$ is the receiver of $\ell_i$
     \ENDFOR
     \STATE Scale powers to take care of noise.
\end{algorithmic}
\label{alg2fig}
\end{algorithm}

\section{Proof of Lemma \ref{ddim}: Covering by circles}
\label{app:b}
\noindent \textbf{Lemma \ref{ddim}:} \emph{
$C_0$ can be covered by $O(1)$ circles of radius $c_1$
 (where $c_1$ is the constant from Lemma \ref{pointsincircle}). 
The area of the annulus $C_t \setminus C_{t-1}$ can be covered by $O(t)$ circles of radius $c_1$, for $t \geq 1$.
}

\begin{proof}
The first claim follows directly from the fact that the $2$-dimensional space has a finite doubling dimension.
Namely, each unit circle can be covered by $O(1)$ radius-$c_1$ circles.
Thus it suffices to prove that $C_t \setminus C_{t-1}$ can be covered
by $O(t)$ unit circles.

Consider now the circle $C$ concentric with $C_0$ with
radius $t + 0.5$, i.e., in the middle of $C_t$ and $C_{t-1}$. 
The circumference of this circle is clearly contained in
$C_t \setminus C_{t-1}$. Now, place $4 \pi (t + 0.5)$ equidistance
points $P$ on this circle. Since the circumference of $C$ is 
$2 \pi (t + 0.5)$, the distance between consecutive points is $\leq 0.5$. 
Now we
claim that all points in $C_t \setminus C_{t-1}$ are within a distance
$1$ of a point in $P$, thus proving that the unit circles
around points in $P$ cover the whole annulus.

Let $x$ be any point in $C_t \setminus C_{t-1}$. Consider the line
connecting this point to the center of $C_0$.  Assume this line
intersects $C$ at point $y$. Now clearly $\|x - y\| \leq 0.5$. On the
other hand, there exists a $p \in P$ such that $\|y - p\| \leq
0.5$. By the triangle inequality $\|x - p\| \leq 1$, completing the
proof.
\end{proof}

\end{document}